\documentclass[twocolumn,pra,aps,amssymb]{revtex4}
\usepackage{amsmath,amssymb,amsfonts,amsthm}
\usepackage{multirow}
\usepackage{verbatim}
\usepackage{url}
\usepackage{comment}
\usepackage{hyperref}
\usepackage{xcolor}
\usepackage{mdframed}

\begin{document}
\newcommand{\ket}[1]{\ensuremath{\left|#1\right\rangle}}
\newcommand{\bra}[1]{\ensuremath{\left\langle#1\right|}}
\newtheorem{definition}{Definition} 
\newtheorem{theorem}{Theorem}

\title{Proposal for Quantum Ciphertext-Policy Attribute-Based Encryption}
\author{Asmita Samanta$^{1}$\footnote {asmita.samanta@tcgcrest.org}, Arpita Maitra$^{1}$\footnote{arpita76b@gmail.com} and Shion Samadder Chaudhury$^1$\footnote{chaudhury.shion@gmail.com}} 
\affiliation{$^1$Institute for Advancing Intelligence (IAI),\\
TCG CREST (Centres for Research and Education in Science and Technology), India.\\}

\begin{abstract}

A Quantum Ciphertext-Policy Attribute-Based Encryption scheme (QCP-ABE) has been presented. In classical domain, most of the popular ABE schemes are based on the hardness of the Bilinear Diffie-Hellman Exponent problem, which has been proven to be vulnerable against Shor's algorithm. Recently, some quantum safe ABE schemes have been proposed exploiting the Lattice problem.  However, no efficient Quantum Attribute-Based Encryption scheme has been reported till date. In this backdrop, in the present initiative, we propose a quantum CP-ABE scheme exploiting Quantum Key Distribution (QKD) and Quantum Error Correcting code. A Semi Quantum version of the scheme has also been considered. Finally, we introduced dynamic access structure in our proposed protocols. 
\end{abstract}
\maketitle

\section{Introduction}

We are currently living in a digital world. Here we have to store and handle tons and tons of data. If we store the huge amounts of data generated every day in our own computers then storage capacity becomes a big issue (especially for large organizations). That is why cloud storing and cloud computing have become very trendy today. 


Cloud is the accumulation of many servers which can be accessed over the internet. These cloud servers are located in data centres all over the world. Cloud computing helps the users and the organizations to use this data just by running some software applications on their own systems.


For Cloud Computing, the main issue is cloud storage. This is supposed to be accessible by anyone. However, it is not desirable that any user can access any data. For example, in a hospital, it is expected that only the doctors and nurses can access the medical history of a patient whereas an administration staff should not be given the permit. There must be some restrictions for the users to access certain data. Hence, to maintain data privacy the data owner should encrypt the data before storing it in cloud storage. 

Now the problem is if one authorized person (data consumer) wants to access these data how can he/she do that? That is how can the data consumer decrypt this data? Simple private or public key encryption does not work here. This is because if we use a private key encryption scheme then the data owner has to send the secret decryption key to all the data consumers by a secret channel, or if we use a public key encryption scheme then data consumers have to send public keys corresponding to their secret decryption key to the data owner and data owner has to encrypt a single data with many keys, which is not a very good solution. Moreover, the data owner does not necessarily know the data consumer in advance. 

A solution to the above problem is to use Attribute-Based Encryption (ABE) \cite{sahai,zhang,cpabe,cpabe2,gl,mabe}, where the data owner specifies which users can decrypt which data. It will depend on specific access control policies. Additionally, any communication between the data owner and the data consumer is not required. 

In \cite{sahai}, Amit Sahai and Brent Waters first introduced the idea of Attribute-Based Encryption (ABE). Since then several variants of the ABE scheme have been introduced. In \cite{zhang}, a detailed discussion about the existing and useful variants of ABE, and the comparisons regarding the usability have been presented. In \cite{cpabe}, John Bethencourt, Amit Sahai, and Brent Waters and in \cite{cpabe2}, Brent Waters proposed efficient and secure constructions of ciphertext-policy ABE (CP-ABE). Whereas in paper \cite{gl}, Vipul Goyal, Omkant Pandey, Amit Sahai, and Brent Waters reported efficient and secure construction of key-policy ABE (KP-ABE). In  \cite{mabe}, Melissa Chase presented  Multi-Authority Attribute Based Encryption. In the paper, she nicely explained the evolution of the ABE constructions and the requirements for those evolutions. In KP-ABE, the secret keys of the users are generated based on an access tree that defines the scope for the privilege of the user, and data are encrypted over the set of attributes. In contrast, CP-ABE uses access trees to encrypt data and the secret keys are generated over a set of attributes.

Most of the recent popular ABE schemes are based on the hardness of the Bilinear Diffie-Hellman Exponent problem. It is well proven that such a scheme can be easily broken by Shor's Algorithm~\cite{shor94}. Hence, these well-known popular ABE schemes are not {\em Quantum-Safe} at al. 

There are some lattice-based ABE schemes \cite{la1,la2} which are believed to be {\em Quantum Safe}. However, those are very involved to understand as well as to implement. On the other hand, there is no {\em Quantum CP-ABE} scheme has been reported till date. In this direction, we propose two QCP-ABE schemes; one is fully Quantum whereas another is Semi-Quantum. 

In fully quantum case, we used Quantum Key Distribution (QKD) for encryption and Quantum Error Correcting Code, specially Calberbank-Shor-Steane (CSS) code for secret distribution. In semi quantum case, we consider a classical Linear Secret Sharing Scheme (LSSS) instead of quantum CSS code for secret distribution. For both the schemes, we analyze the security issues. Finally, we introduced a dynamic access structure in our proposed QCP-ABE schemes.

In a secret sharing scheme, a secret is shared among a set of parties (participants, servers, nodes) in such a way so that certain predetermined subsets of parties called \textit{qualified sets} can reconstruct the secret whereas some predetermined subsets of parties called \textit{forbidden sets} do not have any information about the secret. The monotone collection of qualified sets of participants is called an \textit{access structure}. Secret sharing was introduced independently by Shamir \cite{Shamir} and Blakely \cite{Blakely}. 

Classical (traditional) secret sharing schemes assume that the number of participants, as well as the access structure, are fixed. However, many real-life scenarios demand more flexible secret sharing schemes where the dealer can add/delete parties, modify the access structure, renew shares of participants and modify the conditions for accessing the secret. In the literature, such schemes have been termed dynamic \cite{dynamic}, proactive \cite{proactive}, sequential \cite{sequential}, evolving secret sharing schemes \cite{evolving1,evolving2}. With the advent of quantum computation, quantum secret sharing schemes and quantum versions of the above-mentioned variants of secret sharing schemes have been widely studied \cite{qd1,qd2,qd3,qd4,qd5}. Quantum secret sharing schemes come in two variants: 1) protecting classical information (semi-quantum) and 2) protecting quantum information (fully quantum). In \cite{fq} the authors pointed out that sharing a quantum state is more difficult than sharing a classical information and hence fully quantum secret sharing schemes are much fewer than semi-quantum ones.

\par In the present manuscript, the reason behind considering dynamic access structure is to fulfil the need of the real world. The access control policy for an attribute-based encryption is an access structure of the secret sharing scheme. In real-life, this access structure/policy may change with time. Not only that, the conditions for accessing the data might also be changed. Hence a practical quantum attribute-based encryption scheme should be able to handle the change in access policies with time. Classical attribute based schemes with changing access policies have been considered before \cite{dabe1,dabe2,dabe3}. However, to the best of our knowledge, in the quantum paradigm, this is the first effort for introducing dynamic access structure in QCP-ABE schemes.

\medskip
\noindent\textbf{Organization}: The paper is organized as follows: In section \ref{prel} we outline the notations and preliminaries required for our work. Next, we present an overview of the techniques and ideas in section \ref{overview}. This is followed by sections \ref{technical} and \ref{correct} where we present in detail the constructions, proofs, and the parameters of our construction. We introduce dynamic access structure in the motivation towards more practicality in section \ref{application}. We compare our construction with existing schemes in section \ref{discussion} and the paper is concluded in section \ref{conclusion}.



\section{\label{prel}Preliminaries}
In this section, we introduced a list of few abbreviations and notations which have been used in the manuscript and  brief descriptions of some relevant classical schemes.
\subsection{Abbreviations and Notations}
Following is the list of few abbreviations and notations which have been used in the manuscript.
\begin{itemize}
\item AA: Attribute Authority
\item DO: Data Owner
\item DU: Data User
\item CSP: Cloud Service Provider
\item ABE: Attribute-Based Encryption
\item KP-ABE: Key-Policy Attribute-Based Encryption
\item CP-ABE: Ciphertext-Policy Attribute-Based Encryption
\item LSSS: Linear Secret Sharing Scheme
\item LFSR: Linear-Feedback Shift Register
\item QCP-ABE: Quantum Ciphertext-Policy Attribute-Based Encryption
\item $M$: Message bit-string
\item $m$: Length of the bit-string $B$
\item $B$: Bit-string generated by BB84 and used as seed to the LFSR
\item $B'$: Bit-string corresponding to the basis choice
\item $Z$: $\{\ket{0},\ket{1}\}$
\item $X$: $\{\ket{+},\ket{-}\}$
\item $M_i$: $i$-th bit of $M$
\item $B'_i$: $i$-th bit of $B'$
\item $C$: Ciphertext
\item $mt$: Message-tag
\item $\mathcal{A}$: Access Policy used to encrypt $B'$
\item $n$: Number of total attributes
\item $att_k$: $k$-th attribute
\item $S$: Set of secret shares corresponding to all the attributes
\item $msk$: Master Secret Key used by AA to encrypt $S$
\item $es$: Encrypted form of $S$ using $msk$
\item $al$: Attribute-list of an user DU
\item $AS$: Set of secret shares corresponding to attributes present in $al$
\item $C_i$: $i$-th bit of $C$
\item $\Gamma$: Access structure used for Dynamic Access Policy based QCP-ABE
\item $\Gamma'$: Modified form of $\Gamma$
\item $t$: Time
\item $t_c$: Current time
\item $q$: Qualified
\item $f$: Forbidden
\end{itemize}

\subsection{Attribute-Based Encryption (ABE)}

When we store some data in cloud storage, it becomes accessible to anyone. So, it is necessary to ensure that only authorized users can access the data. To ensure this, we can store encrypted data in cloud storage. We can use either the traditional private-key encryption technology or the traditional public-key encryption technology for the encryption. In a private-key encryption-based access control system, when a new data user (DU) enters into the system, the data owner (DO) has to share the secret key with the new DU such that the new DU can access the data. Similarly, in traditional public-key encryption-based access control, the DO has to re-encrypt its data with the public key of the new DU. That is the DO needs to encrypt single data more than once with different DU's public keys. And also, in both these cases, the DO has to have detailed knowledge about the DUs in advance.\\

So obviously, these two access control mechanisms lack flexibility and scalability (flexibility refers to the expressivity of access control policies, and scalability refers to the effect of newly joined data users on the access control system). In this regard, the Attribute-Based Encryption (ABE) technology plays a crucial role in obtaining access control systems with fine granularity and scalability. In ABE, the DO does not need to know the identity of a specific DU before encryption and the flexible attributes are inserted into the ciphertext. Also, DOs have to do nothing when a new DU joins the system. Therefore, an access control system with ABE permits both flexibility and scalability.
Sahai and Waters first introduced the ABE notion, and it has successfully attracted considerable research efforts as it is a potential cryptographic primitive \cite{zhang}. ABE has two categories, ciphertext-policy ABE (CP-ABE) and key-policy ABE (KP-ABE).\\
~\\
In \textbf{CP-ABE}, a user’s attribute secret key is associated with an attribute list, and a ciphertext specifies an access policy that is defined over an attribute universe of the system. A ciphertext can be decrypted by a user if and only if the user’s attribute list matches the ciphertext’s access policy.\\
In \textbf{KP-ABE}, an access policy, which is defined over the system’s attribute universe, is encoded into a user’s attribute secret key and ciphertext is created concerning an attribute list. A ciphertext can be decrypted by a user if and only if the corresponding attribute list matches the access policy associated with the user’s attribute secret key.

\subsection{Ciphertext-Policy Attribute-Based Encryption (CP-ABE)}

In this paper, we are interested in CP-ABE. The basic CP-ABE has four components, namely the cloud service provider (CSP), the attribute authority (AA), the DO, and the DU \cite{zhang}. The main four algorithms which are used in basic classical CP-ABE are as follows:

\begin{enumerate}
\item \textbf{SetUp ($\lambda$) $\rightarrow$ (PK, MK)}: This algorithm is known as the system setup algorithm. The Attribute Authority (AA) runs this algorithm at the beginning. It takes security parameter $\lambda$ as input and produces system public key PK and master key MK.
\item \textbf{KeyGeneration (PK, MK, L) $\rightarrow$ (SK$_L$) }: This algorithm is known as the attribute key generation algorithm, which is also performed by the AA. The AA takes system public key PK, master key MK, and an attribute list L as inputs and generates SK$_L$ as the attribute secret key corresponding to L. Here the attribute list is the set of attributes of a user and SK$_L$ is called the attribute secret key of that user.
\item \textbf{Encryption (PK, M, $\mathbb{A}$) $\rightarrow$ (CT$_{\mathbb{A}}$)}: This algorithm is performed by the data owner (DO). The DO first chooses an access policy $\mathbb{A}$ for the desired message M, and then takes PK, M, and $\mathbb{A}$ as inputs and produces a ciphertext CT$_{\mathbb{A}}$ of message M associated with the access policy $\mathbb{A}$. This ciphertext is stored on the cloud service provider CSP. 
\item \textbf{Decryption (PK, CT$_{\mathbb{A}}$, SK$_L$) $\rightarrow$ (M or $\perp$)}: This algorithm is performed by the data user DU (DU is basically the data consumer). It takes system public key PK, a ciphertext CT$_{\mathbb{A}}$ of M associated with $\mathbb{A}$, and an attribute secret key SK$_L$ corresponding to DU's attribute list L, and returns M if L satisfies the access policy $\mathbb{A}$. Otherwise outputs the error symbol $\perp$ as an indication of failure of decryption.
\end{enumerate}

But to design Quantum CP-ABE we have slightly modified the above functionalities. Though our construction slightly differs from the classical one, it maintains the main property of ABE, i.e. there will be no communication between DO and DU and DU can access a data if and only if he/she satisfies the access structure attached with the ciphertext using its attribute secret key obtained from AA.

\subsection{Access Structure}
In CP-ABE we encrypts the data under some Access Policy. This Access policy is actually an access structure and any consumer can access the data iff it satisfies the access structure. We mainly prefer monotone Access Structure for our model. Now we are going to give the formal definition \cite{cpabe}.\\~\\
\textbf{Access Structure}: We denote $P = \{ P_1 , P_2 , ... , P_T \}$ as a set of parties. A collection $A \subseteq 2^{\{ P_1 , P_2 , ... , P_T\}}$ is monotonic if $\forall A_1, A_2$ : if $A_1 \in A$ and $A_1 \subseteq A_2$, then $A_2 \in A$. An (monotone) access structure is a (monotone) collection $A$ of non-empty subsets of $P = \{ P_1 , P_2 , ... , P_T \}$. That is, $A \subseteq 2^{\{ P_1 , P_2 , ... , P_T\}} - \{0\}$. So, an access structure $A$ is basically a non-empty set of subsets of $P$. If a subset $B$ of $P$ is in $A$, then we said that that subset $B$ is an authorized set. But if $B$ is not in $A$, then we said that subset $B$ is unauthorized.\\
~\\Here in our case these parties are actually the attributes of an user.\\
There are various types of Access Structure like tree, threshold, AND-OR, LSSS etc.

\subsection{Linear Secret Sharing Scheme (LSSS)}

\textbf{Linear Secret Sharing Scheme (LSSS)} \cite{cpabe2}: Let $P$ denote a set of parties, $s \in \mathbb{Z}_p$ be the shared secret. A secret sharing scheme $\Pi$ over $P$ is linear (over $\mathbb{Z}_p$) if it has the following properties:
\begin{enumerate}
\item The shares of $s$ for each party form a vector over $\mathbb{Z}_p$.
\item There is a matrix $W \in \mathbb{Z}_{p}^{l \times n}$ which is called the share-generating matrix for $\Pi$. For all $i = 1, ..., l$, a function $\rho \in \mathcal{F}( [l] \mapsto P)$ associates the row $W_i$ with a party. To generate the shares, we choose a column vector $\vec{v} = ( s , r_2 , ... , r_n )^{T}$, where $r_2 , ..., r_n$ are randomly picked from $\mathbb{Z}_p$, then $W \cdot \vec{v}$ is the vector of l shares of s according to $\Pi$. The share $(W \cdot \vec{v})_i$ belongs to the party $\rho (i)$.
\end{enumerate}

Every linear secret sharing scheme has the following linear reconstruction property : assume that $\mathbb{A}$ is an access structure. $\Pi$ is an LSSS for $\mathbb{A}$. So, $\mathbb{A} = (W, \rho)$. Now assume $S$ denotes an authorized set, i.e. $S$ satisfies the access structure $\mathbb{A}$. Then let $I = \{ i : \rho ( i ) \in S \}$ be the index set of rows whose labels are in $S$. There exist constants $\{ w_i \in \mathbb{Z}_p \}_{i \in I}$ such that: if the shares $\{ \lambda_i = (W \cdot \vec{v})_i \}$ are valid, then we have $\Sigma_{i \in I} w_i \lambda_i = s$. But for unauthorized sets, no such constants exist. This property is very important to decrypt the message.\\

Also, we can transform tree access structure, or AND-OR access structure into threshold access structure and any threshold access structure can be converted into LSSS. Using a simple algorithm from paper \cite{3,4} we can do it efficiently.\\~\\
\textbf{Example}: Here we will use Shamir’s secret sharing scheme and the construction of the following theorem from paper \cite{3}.\\~\\
\textbf{Theorem}: Let $A_1$ and $A_2$ be monotone access structures defined on participant sets $\mathcal{P}_1$ and $\mathcal{P}_2$, realized by LSSS $(W^{(1)} , \rho^{(1)} )$ of size $w_1$ and $(W^{(2)} , \rho^{(2)} )$ of size $w_2$, respectively. Let $P_z \in \mathcal{P}_1$. There exists an LSSS $(W, \rho)$ of size ~$w_1 + (w_2 - 1) \cdot q$~ realizing the access structure $A_1 (P_z \rightarrow A_2 )$, where $q$ is
the number of rows labeled by $P_z$ in $(W^{(1)} , \rho^{(1)} )$.\\ 
~\\
Suppose the access policy is $((A \wedge B) \vee (C \wedge D)) \wedge E$. Let M is the access matrix and L is the vector whose co-ordinates are attributes. Initially we let $W = (1)$ and $L = (((A,B,2),(C,D,2),1),E,2)$.\\
~\\
\textbf{Step - 1}:  
$W = \begin{pmatrix}
1 & 1\\
1 & 2
\end{pmatrix}$ and $L = \begin{pmatrix}
((A,B,2),(C,D,2),1)\\
E
\end{pmatrix}$\\~\\~\\
\textbf{Step - 2}:  
$W = \begin{pmatrix}
1 & 1\\
1 & 1\\
1 & 2
\end{pmatrix}$ and $L = \begin{pmatrix}
(A,B,2)\\
(C,D,2)\\
E
\end{pmatrix}$\\~\\~\\
\textbf{Step - 3}:  
$W = \begin{pmatrix}
1 & 1 & 1\\
1 & 1 & 2\\
1 & 1 & 0\\
1 & 2 & 0
\end{pmatrix}$ and $L = \begin{pmatrix}
A\\
B\\
(C,D,2)\\
E
\end{pmatrix}$\\~\\~\\
\textbf{Step - 4}:  
$W = \begin{pmatrix}
1 & 1 & 1 & 0\\
1 & 1 & 2 & 0\\
1 & 1 & 0 & 1\\
1 & 1 & 0 & 2\\
1 & 2 & 0 & 0
\end{pmatrix}$ and $L = \begin{pmatrix}
A\\
B\\
C\\
D\\
E
\end{pmatrix}$\\~\\~\\
Here $W$ is the generating matrix of LSSS corresponding to the given access policy and $\rho$ maps $i$-th row of $W$ to $i$-th co-ordinate of $L$.\\
Let us assume that our secret is $s$ and we take a vector $\vec{v} = (s, r_1, r_2, r_3)$, where $r_i$'s are the random values.\\
Now we just compute all the secret shares for the attributes $A$, $B$, $C$, $D$ and $E$ one by one.\\
Secret share for $A$ ($\lambda_{A}$) is $(s + r_1 + r_2)$\\
Secret share for $B$ ($\lambda_{B}$) is $(s + r_1 + 2 \cdot r_2)$\\
Secret share for $C$ ($\lambda_{C}$) is $(s + r_1 + r_3)$\\
Secret share for $D$ ($\lambda_{D}$) is $(s + r_1 + 2 \cdot r_3)$\\
Secret share for $E$ ($\lambda_{E}$) is $(s + 2 \cdot r_1)$\\

Let us consider a set of attributes $S = \{ A, B, E \}$.\\
At the time of decryption, we first compute the matrix just following the algorithm of the paper \cite{3}.\\
~\\
\textbf{Step - 1}:  
$W^{'} = \begin{pmatrix}
1 & 1\\
1 & 2
\end{pmatrix}$ and $L^{'} = \begin{pmatrix}
((A,B,2),(C,D,2),1)\\
E
\end{pmatrix}$\\~\\~\\
\textbf{Step - 2}:  
$W^{'} = \begin{pmatrix}
1 & 1\\
1 & 1\\
1 & 2
\end{pmatrix}$ and $L^{'} = \begin{pmatrix}
(A,B,2)\\
(C,D,2)\\
E
\end{pmatrix}$\\~\\~\\
\textbf{Step - 3}:  
$W^{'} = \begin{pmatrix}
1 & 1 & 1\\
1 & 1 & 2\\
1 & 1 & 0\\
1 & 2 & 0
\end{pmatrix}$ and $L^{'} = \begin{pmatrix}
A\\
B\\
(C,D,2)\\
E
\end{pmatrix}$\\~\\~\\
\textbf{Step - 4}:  
$W^{'} = \begin{pmatrix}
1 & 1 & 1\\
1 & 1 & 2\\
1 & 2 & 0
\end{pmatrix}$ and $L^{'} = \begin{pmatrix}
A\\
B\\
E
\end{pmatrix}$\\~\\

Then we compute a vector $\vec{c} = (c_1, c_2, c_3)$ such that $\vec{c} \cdot W^{'} = (1, 0, 0)$.\\
It can be easily computed that $\vec{c} = (4, -2, -1)$ serves our purpose. It also implies that $S$ is an authorized set of attributes.\\

To get the secret just compute $(c_1 \cdot \lambda_{A} + c_2 \cdot \lambda_{B} + c_3 \cdot \lambda_{E})$.\\

$c_1 \cdot \lambda_{A} + c_2 \cdot \lambda_{B} + c_3 \cdot \lambda_{E}$\\
$= 4 \cdot (s + r_1 + r_2) + (-2) \cdot (s + r_1 + 2 \cdot r_2) + (-1) \cdot (s + 2 \cdot r_1)$\\
$= s \cdot (4 - 2 - 1) + r_1 \cdot (4 - 2 - 2) + r_2 \cdot (4 - 4)$\\
$= s$\\

So it is very easy to compute the secret using the linear reconstruction property of LSSS.
\subsection{Dynamic Access-Structure Based Secret Sharing}

For a set of secrets $\mathcal{S}$, a set of parties $P = \{p_{1}, \ldots, p_{n}\}$ and an access structure $\Gamma$ on $P$, a dynamic secret sharing scheme has four components \cite{dss}, 

\begin{itemize}
    \item \textbf{Share}: For a secret $S \in \mathcal{S}$, the dealer outputs shares $sh_{1}, \ldots, sh_{n}$ to be shared among the parties in $P$.
    
    \item \textbf{Addition}: For a new participant $p_{n+1}$, and a subset $R \subset P$, if $R \notin \Gamma$, dealer assigns share $\perp$ to $p_{n+1}$. Otherwise, $R \in \Gamma$ and without secret reconstruction, dealer assigns new share $sh_{n+1}$ to $p_{n+1}$. 
    
    \item \textbf{Reset}: For a new set of parties $P^{*} = \{p_{n+1}, \ldots, p_{r}\}$, dealer forms a modified set of parties $P^{new} = P \cup P^{*}$ and forms a new access structure $\Gamma'$. For a subset $R\subset P^{new}$, if $R \notin \Gamma$, dealer assigns new shares $\perp$ to the parties of $R \cap P^{*}$. Otherwise $R\in \Gamma'$ and dealer assigns  shares $sh_{1}^{new}, \ldots, sh_{|R|}^{new}$ to parties in $R$ without secret reconstruction and the old shares of the parties of $R\cap P$ are deleted. 
    
    \item \textbf{Reconstruct}: The dealer takes shares of a subset of parties $R \subset P(\text{or }P^{new})$ as input and if $R \in \Gamma (\text{or }\Gamma')$, outputs the secret, or outputs $\perp$ if $R  \notin \Gamma (\text{or }\Gamma')$
\end{itemize}

\section{\label{overview}Basic Idea of the Construction and Used Tools}
Before going to the formal description of the proposed algorithms, in this section, we are giving an outline of the idea behind the schemes; Fully Quantum CP-ABE and Semi Quantum CP-ABE. 

For both of our schemes, we have used Quantum Key Distribution (QKD) protocol for encryption of a message. For simplicity, we consider BB84~\cite{bb84} here. However, any QKD protocol is applicable for the proposed schemes. Firstly, we present Semi-Quantum CP-ABE followed by Fully Quantum CP-ABE.

In the first scheme, i.e. for Semi-Quantum scheme, our idea is very simple. We encrypt the message using two non-orthogonal bases $\{\ket{0}, \ket{1}\}$ and $\{\ket{+},\ket{-}\}$, where $\ket{+}=\frac{1}{\sqrt{2}}(\ket{0}+\ket{1})$ and $\ket{-}=\frac{1}{\sqrt{2}}(\ket{0}-\ket{1})$. The security comes from the  ``no-cloning" theorem~\cite{wootters,dieks} and indistinguishability of two non-orthogonal states. A random bit string will be generated to decide the order of the bases for encryption. If the bit is 0, the message bit will be encrypted in $\{\ket{0},\ket{1}\}$ basis. If the bit is 1, the message bit will be encrypted in $\{\ket{+},\ket{-}\}$ basis. Hence, if one can get this sequence of basis choice then he/she can decrypt the actual message bit. 
 In Semi-Quantum case, we take the integer value of this bit pattern and use the classical LSSS scheme to generate the secret shares for that integer. Shares will be distributed amongst the users using the attributes defined by DO for each of them.\\
 
Here the parties of LSSS are the attributes of the users. Attributes are names, designations, departments, job ids, etc. of the users. We distribute the secret between all of these attributes. Now whenever one user asks for its attribute secret key, AA gives him/her the secret shares corresponding to the attributes he/she owns. Then the user uses these secret shares to reconstruct the secret. If the user contains the set of attributes that is a qualifying set in the access structure, then he/she can reconstruct the secret integer, and hence he/she can decrypt the actual message bit string.\\
~\\
For the second scheme, i.e., for the Fully Quantum scheme, we have used a quantum secret sharing scheme based on CSS code~\cite{nc}. Here the bits in the bit string which will decide the encryption bases are considered as quantum logical bits, i.e,  bit $0$ as quantum bit $\ket{0}_L$ and bit $1$ as quantum bit $\ket{1}_L$ (Here $\ket{0}_L$ and $\ket{1}_L$ are generated using quantum CSS code). Then we generate secret shares of $\ket{0}_L$ and $\ket{1}_L$ using the quantum secret sharing scheme~ \cite{maitra} and gave the attribute secret key as follows:\\
If the number of attributes in the system is $n$ and the number of elements (bits) in the bit string is $m$ then secret share collection is
$S := \{\{$secret share of the qubit $\ket{j}_L$ corresponding to the k-th attribute $att_k$ $\}_k\}_i$ where  $i \in \{1,2,...,m\}$, $k \in \{1,2,...,n\}$ and $j \in \{0,1\}$.\\
Whenever an user asks for its attribute secret key then AA sends (whose attribute list is $al := \{att_{i_j}\}_{j \in [n']}$ where $n' < n$) $AS := \{\{$secret share of the qubit $\ket{j}_L$ corresponding to the k-th attribute $att_k$ $\}_{att_k \in al}\}_i$ (where in the sequence $i$-th basis choice is $j$, $i \in \{1,2,...,m\}$, $k \in \{1,2,...,n\}$ and $j \in \{0,1\}$) to the user. Then user use this to reconstruct each $m$ bit of the sequence of basis choice and hence he/she can easily decrypt the actual message string.

 For example, let the bit string which decides the encryption bases for a certain message $M$ be $01100011$. Now, consider the seven qubits Stean Code~\cite{nc}. The code is constructed from the two classical linear codes such that $\{0\}\subset {C}_{1} \subset C \subset \mathbb {F}_{2}^{7}$ with the generator matrices 

$$ G =
\begin{pmatrix}
1&0&0&0&0&1&1\\
0&1&0&0&1&0&1\\
0&0&1&0&1&1&0\\
0&0&0&1&1&1&1
\end{pmatrix},
$$

$$ G_1 =
\begin{pmatrix}
0&0&0&1&1&1&1\\
0&1&1&0&0&1&1\\
1&0&1&0&1&0&1
\end{pmatrix}.
$$
From this we get the following.
\begin{eqnarray*}
\ket{0}_L=\frac{1}{\sqrt{8}}[\ket{0000000} + \ket{0101101} + \ket{0011011} + \ket{0110110} \\
+ \ket{1111000} + \ket{1010101} + \ket{1100011} + \ket{1001110}]
\end{eqnarray*}
and, \begin{eqnarray*}
\ket{1}_L=\frac{1}{\sqrt{8}}[\ket{1111111}  + \ket{1010010} + \ket {1100100} +
\ket{1001001}\\
+  \ket{0000111}+ \ket{0101010} + \ket{0011100} + \ket{0110001}]
\end{eqnarray*}
It is easy to see that if we measure second, fourth and sixth qubits both for $\ket{0}_L$ and $\ket{1}_L$ and XORED the values, we will get back the logical bit values, i.e., for $\ket{0}_L$ 0 and for $\ket{1}_L$, 1. Hence, we can say that $\{2,4,6\}$ is a valid access structure for the above code. 

Now, according to the protocol the bit string will be $0_L1_L1_L0_L0_L0_L1_L1_L$. 
Let a user $A$ has attributes $\{2,4,6\}$. Each integer stands for an attribute. For example, 2 stands for name, 4 for designation and 6 for institute. Let A be authorized for the message $M$. In this case, AA will share the qubits $2,4,6$ for each $0_L$ and $1_L$. After getting the shares, $A$ will measure the qubits in $\{\ket{0},\ket{1}\}$ basis and xor the output bits. The XORED bit is the secret bit of the bit string. 

One should note that the set of attributes for a user might be a super set of the access structure of the secret sharing scheme. In that case, AA need not to send the shares corresponding to all attributes  to the user. It is enough to send the shares which satisfy the access structure of the scheme. In other words, the set of minimal attributes which satisfies the access structure of the scheme are sufficient to reconstruct the secret. In the example, if the attributes of A are $2, 4, 6, 5, 7$, then AA will send the shares corresponding to the attributes $2, 4, 6$ only.
\section{\label{technical}Technical Details}
\subsection{Semi Quantum Ciphertext-Policy Attribute-Based Encryption}


\textbf{\underline{System Setup}}:
\begin{enumerate}
\item DO chooses its message $M \in \{0,1\}^{2^m}$ and then express $M$ in bits (classical 0-1 bit)
\item AA chooses two bases $Z := \{\ket{0}, \ket{1} \}$ and $X := \{\ket{+}, \ket{-}\}$. AA makes this choice of bases public and also broadcast that whenever user will get 0, he/she will use Z basis and will X basis if gets 1. AA also selects a master secret key $msk$ for its own use.
\end{enumerate}
\textbf{\underline{Encryption}}:
\begin{enumerate}
\item To encrypt $M$, DO performs BB84 with AA and shares a bit-string $B$ of length $m$. Here, we consider $m$ bit security for the encryption scheme.
\item AA also generates a message tag $mt$ which contains the unique id of the DO and the time of performing BB84. AA sends this message tag $mt$ to the DO.
\item DO uses a secure stream cipher~\cite{SoK} to generate 
 to get a bit-string $B'$ of length $2^{m}$ using $B$ as seed.
\item DO take the bit expression of $M$ and consider the bit string $B'$ as basis choice. Then he generate the ciphertext as follows:
\begin{enumerate}
\item If $M_i = 0$ and $B'_i = 0$, DO generates the qubit $\ket{0}$.
\item If $M_i = 1$ and $B'_i = 0$, DO generates the qubit $\ket{1}$.
\item If $M_i = 0$ and $B'_i = 1$, DO generates the qubit $\ket{+}$.
\item If $M_i = 1$ and $B'_i = 1$, DO generates the qubit $\ket{-}$.
\end{enumerate}
Then DO sets this qubit-string as the cipher text $C$ of $M$.
\item DO choose an access policy $\mathcal{A}$ for this cipher text $C$, which specifies who can access this data and who can not.
\item DO sends this access policy $\mathcal{A}$ along with the message tag $mt$ to AA and stores the pair $(mt,~C,~\mathcal{A})$ in Cloud Storage.
\end{enumerate}
\textbf{\underline{Key Generation}}:
\begin{enumerate}
\item Let the total number of attributes in the system be $n$. AA takes $B$ which they generated by BB84. Then AA consider the integer $Int$ whose bit expression is same as $B$.
\item Then AA uses LSSS to generate secret shares of the secret $Int$ corresponding to each $n$ attributes. Then AA set the sequence $S := \{$secret share of $Int$ corresponding to the k-th attribute $att_k$ $\}_k$ where $k \in \{1,2,...,n\}$.
\item As we are considering $m$-bit security,  AA use $AES_{2m}$ to encrypt $S$ using the master secret key $msk$ and get $es$ as output. AA generates a pair $(mt,~es)$ and stores this in Cloud Storage.
\end{enumerate}
\textbf{\underline{Decryption}}:
\begin{enumerate}
\item An authorized DU first get the pair $(mt,~C,~\mathcal{A})$ from Cloud Storage and send $mt$ to AA along with its own attribute list $al := \{att_{i_j}\}_{j \in [n']}$ where $n' < n$.
\item AA get the pair $(mt,~es)$ from Cloud Storage and decrypt $es$ using $msk$ to get $S$.
\item Then AA sends $AS := \{$secret share of $Int$ corresponding to the k-th attribute $att_k$ $\}_{att_k \in al}$ (where $k \in \{1,2,...,n\}$) to DU.
\item DU computes the following:
\begin{enumerate}
\item DU uses reconstruction property of LSSS to get back the integer $Int$. After that DU sets $B$ as the bit expression (of length $m$) of $Int$.
\item Then DU uses the secure stream cipher~\cite{SoK} to generate $B'$ of length $2^{m}$ using $B$ as seed.
\item Then DU measures $C_i$ in Z basis if $B'_i = 0$ and measures $C_i$ in X basis if $B'_i = 1$.
\item Clearly after the above measurement DU gets the bit expression of the actual message $M$ and hence get $M$. 
\end{enumerate}
\end{enumerate}

\subsection{Fully Quantum Ciphertext-Policy Attribute-Based Encryption}

\textbf{\underline{System Setup}}: (same as semi-quantum)
\begin{enumerate}
\item DO chooses its message $M \in \{0,1\}^{2^m}$ and then express $M$ in bits (classical 0-1 bit)
\item AA chooses two bases $Z := \{\ket{0}, \ket{1} \}$ and $X := \{\ket{+}, \ket{-}\}$. AA makes this choice of bases public and also broadcast that whenever user will get 0, he/she will use Z basis and will X basis if gets 1. AA also selects a master secret key $msk$ for its own use.
\end{enumerate}
\textbf{\underline{Encryption}} (same as semi-quantum):
\begin{enumerate}
\item To encrypt $M$, DO performs BB84 with AA and shares a bit-string $B$ of length $m$. Here, we consider $m$ bit security for the encryption scheme.
\item AA also generates a message tag $mt$ which contains the unique id of the DO and the time of performing BB84. AA sends this message tag $mt$ to the DO.
\item DO uses a secure stream cipher~\cite{SoK} to generate 
 to get a bit-string $B'$ of length $2^{m}$ using $B$ as seed.
\item DO take the bit expression of $M$ and consider the bit string $B'$ as basis choices. Then he generate the ciphertext as follows:
\begin{enumerate}
\item If $M_i = 0$ and $B'_i = 0$, DO generates the qubit $\ket{0}$.
\item If $M_i = 1$ and $B'_i = 0$, DO generates the qubit $\ket{1}$.
\item If $M_i = 0$ and $B'_i = 1$, DO generates the qubit $\ket{+}$.
\item If $M_i = 1$ and $B'_i = 1$, DO generates the qubit $\ket{-}$.
\end{enumerate}
Then DO sets this qubit-string as the cipher text $C$ of $M$. DO stores the pair $(mt,~C)$ in Cloud Storage.
\item DO choose an access policy $\mathcal{A}$ for this cipher text $C$, which specifies who can access this data and who can not.
\item DO sends this access policy $\mathcal{A}$ along with the message tag $mt$ to AA.
\end{enumerate}
\textbf{\underline{Key Generation}}:
\begin{enumerate}
\item Let the total number of attributes in the system be $n$. AA takes $B$ which they generate by BB84 and does the following:
\begin{enumerate}
\item If $B_i = 0$, AA consider the secret qubit at $i$-th position as $\ket{0}_L$ (Here, we consider Quantum CSS code with $n$ qubits).
\item If $B_i = 1$, AA consider the secret qubit at $i$-th position as $\ket{1}_L$ (Here, we consider Quantum CSS code with $n$ qubits).
\end{enumerate}
\item Then AA generates secret shares of $\ket{i}_L$ corresponding to each $n$ attributes (considering the attributes as the parties of the quantum secret sharing scheme), for $i = 0,1$.
\item Then AA generates the sequence $S := \{\{$secret share of the qubit $\ket{j}_L$ corresponding to the k-th attribute $att_k$ $\}_k\}_i$ where $B_i = j$, $i \in \{1,2,...,m\}$, $k \in \{1,2,...,n\}$ and $j \in \{0,1\}$.
\item As we are considering $m$-bit security, AA use $AES_{2m}$ to encrypt $S$ using the master secret key $msk$ and get $es$ as output. AA generates a pair $(mt,~es)$ and stores this in Cloud Storage.
\end{enumerate}
\textbf{\underline{Decryption}}:
\begin{enumerate}
\item An authorized DU first get the pair $(mt,~C)$ from Cloud Storage and send $mt$ to AA along with its own attribute list $al := \{att_{i_j}\}_{j \in [n']}$ where $n' < n$.
\item AA get the pair $(mt,~es)$ from Cloud Storage and decrypt $es$ using $msk$ to get $S$.
\item Then AA sends $AS := \{\{$secret share of the qubit $\ket{j}_L$ corresponding to the k-th attribute $att_k$ $\}_{att_k \in al}\}_i$ (where $B_i = j$, $i \in \{1,2,...,m\}$, $k \in \{1,2,...,n\}$ and $j \in \{0,1\}$) to DU.
\item DU computes the following :
\begin{enumerate}
\item For each $i \in \{1,...,m\}$, reconstruct $B_i$ from $AS_i :=\{$secret share of the qubit $\ket{j}_L$ corresponding to the k-th attribute $att_k$ $\}_{att_k \in al}$ (using reconstruction part of quantum secret sharing, i.e., measure the superposition of all of the shares and then XOR all the bits of the measured state).
\item Then DU uses the secure stream cipher~\cite{SoK} to generate $B'$ of length $2^{m}$ using $B$ as seed.
\item Then DU measures $C_i$ in Z basis if $B'_i = 0$ and measures $C_i$ in X basis if $B'_i = 1$.
\item Clearly after the above measurement DU gets the bit expression of the actual message $M$ and hence get $M$. 
\end{enumerate}
\end{enumerate}

\section{\label{correct}Security Analysis of the Proposed QCP-ABE Schemes} In this section, we prove the correctness and secrecy of the schemes.
\subsection{Correctness}
\begin{theorem}
The proposed QCP-ABE schemes (Semi Quantum and Fully Quantum) are correct.
\end{theorem}
\begin{proof}

 To prove the correctness of the scheme, first we have to show that for any message $M \in \{0,1\}^{2^m}$, Decryption(Encryption($M,B'_1$),$B'_2$) = $M$ (where $B'_1$ is the sequence used for the basis selection  at the time of encryption and $B'_2$ is the sequence used for the basis selection  at the time of decryption).

As we know that if a  qubit is generated and measured in the same basis then we can undo the measurement, i.e., the qubit remains as it is. For example, if we generate a qubit $q$ in $\{\ket{0},\ket{1}\}$ basis and measure in the same basis, then with probability 1 we can know whether $q$ was $\ket{0}$ or $\ket{1}$. Similarly, if we generate a qubit $q$ in $\{\ket{+},\ket{-}\}$  measure it in the same basis, then with probability 1 we can know whether $q$ was $\ket{+}$ or $\ket{-}$.

So, to prove the correctness it is enough to show that $B'_1 = B'_2$. Now if we observe our construction, we can see that we have generated the $2^m$-length string $B'_i$ using $m$-length string $B_i$ as a ``seed'' to a pseudo-random number generator, for $i \in \{1,2\}$. We know that if we feed the same ``seed'' as an input to the same pseudo-random generator at both encryption and decryption end, it will generate the same bit-string every time. 
So  we have to prove that $B_1 = B_2$.\\

For our first construction, i.e, for {\em Semi Quantum} case, we consider the seed $B_i$ as an integer $Int_i$ ( for $i \in \{1,2\}$). Here $Int_2$ is reconstructed from the secret shares of $Int_1$. To generate the secret share we have used LSSS. So $Int_1 \neq Int_2$ contradicts the correctness of well-known secret sharing scheme LSSS. Therefore, we get $Int_1 = Int_2$ which implies $B_1 = B_2$. Hence we get $B'_1 = B'_2$.

Similarly, for our second scheme, i.e., in case of {\em Fully Quantum}, the correctness of quantum secret sharing scheme leads to the result that $B_1 = B_2$. Hence we get $B'_1 = B'_2$.
This completes the proof.
\end{proof}

\begin{theorem}
A data user,  DU can successfully decrypt an encrypted message if and only if he/she satisfies the attributes determined by the data owner DO.
\end{theorem}

\begin{proof}

First we will show that if the shares holding by a DU satisfies the access structure then DU can successfully decrypt an encrypted message. 
For this, it is enough to prove that if the attributes of a DU satisfy the access structure then DU can reconstruct the secret $B$. 

According to the protocol, both for Semi-Quantum and Fully Quantum, the shares corresponding to the attributes will be distributed among the users. Thus, if the attributes for a user satisfy the access structure of the secret sharing scheme, DU can reconstruct the secret.
For our first scheme, i.e., for {\em Semi Quantum} scheme correctness of LSSS and for our second scheme, i.e., for {\em Fully quantum } case, the correctness of quantum secret sharing scheme guarantees that if DU satisfies the access structure then DU can reconstruct the secret $B$.

For the converse part, let there exist a user who does not satisfy the access structure but can successfully decrypt the ciphertext. In other words, user's attributes do not satisfy the access structure but with the secret shares corresponding to those attributes he/she can reconstruct the secret, which violates the correctness of the secret sharing scheme for  both the cases. Therefore a data user DU can successfully decrypt an encrypted message if and only if he/she satisfies the attributes determined by the data owner DO.
\end{proof}
\subsection{Security}

For our constructions, we have some basic security assumptions like the classical one. We have assumed that 
\begin{enumerate}
\item the Attribute Authority (AA) is fully trusted.
\item Attribute Authority (AA) uses secure communication channels to send the secret shares of the key to the Data Users (DU).
\item an adversary can only interfere in the channel at the time of QKD. 
\end{enumerate}
Additionally, at the time of BB84, QKD amongst DO and AA, we have considered the followings assumptions.
\begin{enumerate}
    \item We assume the inherent physical correctness of 
Quantum Mechanics. 
\item No information leakage takes place from the legitimate parties' (DO and AA) laboratories. 
\item  The legitimate parties (DO and AA) have a sufficiently good knowledge about their sources.
    \end{enumerate}
    Based on these assumptions, we came up with the following theorems.
\begin{theorem}
The proposed QCP-ABE schemes are $\epsilon$-secure, where $\epsilon$ is infinitesimally small.
\end{theorem}
\begin{proof}
According to the QCP-ABE protocols, Semi Quantum and Fully Quantum, the initial bit-string $B$ of length $m$ has been shared between DO and AA via BB84 (variants of BB84 can also be used). From the security proof of BB84, it is well proven that under the above assumptions the protocol is $\epsilon$ secure, where $\epsilon$ is infinitesimally small. Thus, $B$ is also $\epsilon$ secure. That is $\Pr(B=B_{adv})=\epsilon$, where $B_{adv}$ is the string generated at the adversary's end.

Next, $B$ will be used as a ``seed'' for the secure stream cipher used to generate $B'$ at both the encryption and decryption ends, i.e. at the ends of DO and DU. Hence, the proposed QCP-ABE schemes are $\epsilon$-secure.
\end{proof}
\begin{theorem}
The proposed QCP-ABE schemes demand $m$-bit security.
\end{theorem}
\begin{proof}
Consider, {\em step 3 of Key Generation phase} for both the protocols. In this phase, AA use $AES_{2m}$ to encrypt $S$ using the master secret key $msk$ and get $es$ as output. AA generates a pair $(mt,~es)$ and stores this in Cloud Storage. To decrypt $S$ from the cloud storage, the adversary needs to extract the secret key of $AES_{2m}$. In this case, Grover search over the key space will provide the optimal solution in $O\sqrt(2^{2m})=O(2^{m})$. Thus, the proposed QCP-ABE schemes demand $m$-bit security.
\end{proof}
\section{\label{application}Quantum Attribute-based encryption with dynamic access policies} To include dynamic features in the constructed scheme above we consider the following :
\begin{enumerate}
    \item Expandability: A new user should be able to enter the system.
    
    \item Renewability: The private key, attribute set and attribute values can be renewed and the old private key should not be of any use after the parameters of the user are renewed.
    
    \item Revocability: A user's private key can be revoked and the old keys should not decrypt the ciphertexts.
    
    \item Independence: When attribute updating occurs, the existing users are not required to renew their private keys. 
\end{enumerate}

\noindent Note that since we are working in the quantum domain, a very useful tool is the ``no-cloning" theorem and by this theorem, a user cannot make a copy of its share or hold on to its share. This simplifies the implementation of the above points in the scheme to a large extent. Let us suppose that the attribute authority AA now has a dynamic (evolving) access structure.

\medskip
    \textbf{Procedure 1:} \textbf{\underline {Modifications by AA.}}
    \begin{enumerate}
    
        \item {Do Step 2 for all qualified and forbidden sets at time $t$.}
        \item {Add a time-stamp to the access structure $\Gamma$ for time $t$, i.e., a qualified (resp. forbidden) set of attributes $(a_{1}, \ldots, a_{k})$ is modified as $(a_{1}, \ldots, a_{k}, q, t)$ (resp $(a_{1}, \ldots, a_{k}, f, t)$) where $q$ (resp. $f$) denotes that it is qualified (resp forbidden) at time $t$.}
        
        \item {Suppose the access structure  changes at time $t+1$ in two ways, either addition of new attribute(s) or change in access structure $\Gamma$ to a new access structure $\Gamma'.$} 
        \item In case of the addition of new attribute(s), no modification is required as the minimal set of attributes are sufficient to reconstruct the secret (please find section III, last paragraph). However, the total number of attributes is now $n+1$.
        \item If the attribute(s) changes in such a way so that the access structure $\Gamma$ changes to $\Gamma'$, then do the following.
        \item Do step 7 to 11 for all users at time $t+1$.
        \item {If a set $(a_{1}, \ldots, a_{k}, q, t) \in \Gamma$ becomes a forbidden set in $\Gamma'$, modify it as $(a_{1}, \ldots, a_{k}, f, t+1)$.}
        
        \item {If a set $(a_{1}, \ldots, a_{k}, f, t) \in \Gamma$ becomes a qualified set in $\Gamma'$, modify it as $(a_{1}, \ldots, a_{k}, q, t+1)$.}
        
        \item {If a set $(a_{1}, \ldots, a_{k}, q, t) \in \Gamma$ remains  qualified in $\Gamma'$, no changes are needed.}
        
        \item {If a set $(a_{1}, \ldots, a_{k}, f, t) \in \Gamma$ remains forbidden in $\Gamma'$, no changes are needed.}
        
        \item {For a set $Q \in \Gamma'$ such that $Q \notin \Gamma$, add time stamp as $(Q, q, t+1)$.}
        
        \item Stop.
    \end{enumerate}

\medskip
    \textbf{Procedure 2:} \textbf{\underline{Modifications for DU.}}
    \begin{enumerate}
    \item {An authorized DU first gets the pair $(mt,~C,~\mathcal{A})$ from Cloud Storage and send $mt$ to AA alongwith its own attribute list $al := \{att_{i_j}\}_{j \in [n']}$.}
    
    \item {AA/DU adds a time-stamp $t$ for comparison with the existing attribute policy.}
    
    \item Stop.
    \end{enumerate}

\medskip
    \textbf{Procedure 3:} \textbf{\underline {Modifications by DO when access}}\\
    \textbf{\underline {structure changes}}
    
    \begin{enumerate}
    \item {Let current time $= t_{c}.$}
    \item {Evaluate all previously submitted attributes for time $< t_{c}.$} 
    \item {For a time stamped set of attributes $(a_{1}, \ldots, a_{k}, t)$ ($t < t_{c}$), sent it to AA.}
    \item {AA runs procedure 1 and returns current status of the attributes ($q$ or $f$).}
    
    \item {If the current state is $q$, no modification needed.}
    
    \item {If the current state is $f$, assign a random state $\ket{\perp}$ to DU.}
    
    \item {If the current state is $f$, but previously it was $q$, give a random permutation to the encrypted states, i..e., states in $C$. This can also be done by the quantum one-time pad methodology, i.e., one key for one time only. And assign a random state $\ket{\perp}$ to DU. }
    
    \item {If the current state  $q$, but previously it was $f$, run encryption procedure of section IV.}
    
    \item Stop.
    \end{enumerate}

\medskip
\noindent Note that in section \ref{technical}, we have not assumed any \textit{apriori} bound on the number of data users. Any user with the correct set of attributes (as per the current access policy) is able to decode the secret. Hence from the point of the number of data users, our constructed scheme is dynamic.  Moreover,  the modifications of Procedures 1, 2 and 3, make the scheme renewable due to the following reason:  if at the current time, the attributes of the DU fails to be a qualified attribute policy, the DO/AA may assign a random state to the DU, i.e., applying a random permutation to the encrypted states, i..e., states in $C$. This can also be done by one-time pad methodology, i.e., one key for one time only. Again due to the ``no-cloning" theorem and indistinguishability of non-orthogonal states, a party cannot get any information about the old state from the new random state without the new secret key. The DU gets a new state by the expandable property of our constructed schemes. The DO/AA  revokes the private key of the DU which satisfy the revocability property of the schemes. Finally using the steps 7 and 8 of Procedure 1, the existing users whose attributes remain qualified at a later time do not need to renew their private keys. Hence our modifications make the scheme expandable, renewable, revocable and independent.

\section{\label{discussion}Discussions} In Ciphertext-Policy Attribute-Based Encryption, it is desirable that any party who has a correct set of attributes can reconstruct the secret. In this direction, we have made our encryption scheme dynamic. However not every dynamic scheme can handle unbounded number of participants. Specialized dynamic schemes which can handle such scenarios are evolving secret sharing schemes \cite{evolving2,qd5}. Our constructed scheme is simple and flexible enough to incorporate such schemes in the encryption method. Another advantage of  our scheme is that the dimension of the share states is exactly equal to the secret state (qubit of dimension 2) and hence is  within the reach of practical implementation. To compare our scheme with the existing schemes, to the best of our knowledge, this is the first construction of a quantum attribute-based encryption scheme.
We have proposed two schemes. One scheme is fully quantum and another is semi quantum. 

\section{\label{conclusion}Conclusion}

In the domain of classical cryptography, extensive work has been done on ABE. However, most of them depend on the hardness of  discrete logarithm problem which are proven to be vulnerable against Shor's algorithm. Recently, some Lattice based CP-ABE schemes have been proposed, but those are not as efficient as traditional Public Key Cryptography (RSA and ECC) and hard to implement. In this backdrop, we present two QCP-ABE schemes based on Quantum Key Distribution (QKD) and Quantum Error Correction Code (QECC). 

In this manuscript, we took help of CSS code only. However, it can be extended for any other QECC. In semi quantum version, we use Classical Linear  Secret Sharing Scheme (LSSS), though the encryption part remains same as fully quantum one, i.e.,  based on QKD. And that is why we call it semi quantum. Finally, the idea of dynamic access structure has been introduced to make the schemes more realistic. 

One limitation of our scheme is that one data user (DU) can decrypt only one message with one secret attribute key. For each decryption, data user (DU) has to ask attribute authority (AA) for the attribute key, which increases the communication cost. So one future work in this direction may be to modify this scheme in such a way so that data users (DU) can decrypt multiple messages using one attribute key. Another future work would be to extend our methods for constructing a quantum functional encryption scheme which is a generalization of an Attribute-Based Encryption scheme.

\end{document}